\documentclass[12pt]{article}
\usepackage{amssymb}
\usepackage{amsmath,amsthm}
\usepackage[latin1]{inputenc}
\usepackage{graphicx}
\usepackage{hyperref}
\usepackage{enumerate}
\usepackage{tikz}
\usepackage{tkz-graph}
\usetikzlibrary{plotmarks}
\usetikzlibrary{shadows}
\usetikzlibrary{patterns}
\usetikzlibrary{calc}
\usepackage{mathrsfs}
\usepackage{verbatim}
\usepackage{algorithm}
\usepackage{algorithmic}

\usepackage{subfig}
\newcommand{\subsubfloat}[2]{%
  \begin{tabular}{@{}c@{}}#1\\#2\end{tabular}%
}
\usetikzlibrary{arrows, decorations.markings}

\makeatletter
\newcommand\appendix@section[1]{%
\refstepcounter{section}%
\orig@section*{Appendix \@Alph\c@section: #1}%
}
\let\orig@section\section
\g@addto@macro\appendix{\let\section\appendix@section}
\makeatother

\makeatletter
\newcommand{\algorithmicbreak}{\textbf{break}}
\newcommand{\BREAK}{\STATE \algorithmicbreak}
\makeatother

\hypersetup{colorlinks=true, linkcolor=blue, citecolor=blue, urlcolor=blue}

\newtheorem{remark}{Remark}

\newtheorem{definition}[remark]{Definition}

\newtheorem{theorem}[remark]{Theorem}
\newtheorem{proposition}[remark]{Proposition}

\newcommand{\st}{\ensuremath{~\textnormal{s.t.}~}}

\newcommand{\boundedDist}{bounded distance}
\newcommand{\BoundedDist}{Bounded distance}
\newcommand{\boundedRepre}{bounded representation}
\newcommand{\adjacencyRepre}{adjacency representation}

\DeclareMathOperator{\adim}{adim}

\definecolor{fuchsia}{HTML}{800080}

\newcommand\RT[1]{\textcolor{blue}{RT: #1}}
\newcommand\YR[1]{\textcolor{red}{YR: #1}}
\newcommand\SM[1]{\textcolor{orange}{SM: #1}}
\newcommand\TODO[1]{\textcolor{red}{TODO: #1}}

\newcommand\loss{\operatorname{loss}}

\title{Rethinking $(k,\ell)$-anonymity in social graphs: $(k,\ell)$-adjacency anonymity and $(k,\ell)$-(adjacency) anonymous transformations}

\author{S. Mauw$^1$, Y.  Ram\'{i}rez-Cruz$^2$ and R. Trujillo-Rasua$^2$\\ 
{\small $^1$CSC, $^2$SnT}\\\small{University of Luxembourg}\\ 
{\small 6, av. de la Fonte, L-4364 Esch-sur-Alzette, Luxembourg}\\ 
{\small \{sjouke.mauw, yunior.ramirez, rolando.trujillo\}\@@uni.lu}} 

\begin{document}
\maketitle

\begin{abstract}
This paper treats the privacy-preserving publication of social graphs in the 
presence of active adversaries, that is, adversaries with the ability to 
introduce sybil nodes in the graph prior to publication and leverage them to 
create unique fingerprints for a set of victim nodes and re-identify them 
after publication. Stemming from the notion of $(k,\ell)$-anonymity, 
we introduce 
$(k,\ell)$-anonymous transformations, characterising graph perturbation methods 
that ensure protection from active adversaries levaraging up to $\ell$ 
sybil nodes. Additionally, we introduce a new privacy property: 
$(k,\ell)$-adjacency anonymity, which relaxes the assumption made 
by $(k,\ell)$-anonymity that adversaries can control all distances 
between sybil nodes and the rest of the nodes in the graph. The new 
privacy property is in turn the basis for a new type of graph perturbation: 
$(k,\ell)$-adjacency anonymous transformations. We propose algorithms for 
obtaining $(k,1)$-adjacency anonymous transformations for arbitrary values 
of $k$, as well as $(2,\ell)$-adjacency anonymous transformations 
for small values of $\ell$. 
\end{abstract}


\section{Introduction}
\label{sectIntro}

Online social networks (OSNs) have become the most successful application 
of our time. Nearly two billion persons\footnote{Source: 
\href{https://www.statista.com/statistics/272014/global-social-networks-ranked-
by-number-of-users/}{statista.com}, consulted on April 21st, 2017.} regularly use some OSN to 
interact with 
friends and relatives, share information, get news, entertainment, etc. 
As a result of this, massive amounts of information about human behaviour, 
personal relationships, consumption patters, personal preferences, and more, 
are generated everyday. 
An important part of this information is encoded in the form of social graphs. 
In a social graph, every vertex corresponds to a person (a user 
of the OSN), whereas edges represent relations between 
individuals. Rich personal information, such as name, address, etc. is usually 
associated to vertices as attributes. Edges can also have associated attributes, 
which may encode, for instance, 
the nature of the relation (friendship, laboral, family), the date and place 
where it was established, etc. 

This massive amount of information is enormously valuable. OSNs themselves 
analyse this data in order to determine the advertisement they show to their 
users, suggest new potential friends, filter out the information that they 
consider not to be interesting to the user, etc. As holders of the information, 
the OSN can effectively access the totality of the available data, as authorised 
by the users when they upload their information. However, third parties, 
such as social scientists, market researchers, public institutions 
and private companies, are also interested in accessing and analysing 
a part of this information for conducting population studies, assessing 
the effect of communication campaigns, surveying public opinion, and 
many other purposes.  
In order to enable these studies, it is necessary that the OSN administrators 
release samples of their social graphs. However, despite the usefulness 
of the studies that can be conducted on the released data, the sensitive 
nature of a part of the information encoded in social graphs, e.g. political 
or religious affiliation, arises serious privacy concerns. 

A \emph{na\"ive} approach to protect the privacy of users 
in publishing social graphs is to remove all personally 
identifiable information from the released graph. However, as shown 
in \cite{NS2009}, even a graph with no identifying attributes can leak 
sensitive information, since some structural properties (the degree of 
vertices, their neighbourhoods, etc.) can be unique to certain users.  
A \emph{re-identification attack} seeks to leverage some knowledge about 
a set of users, the victims, to re-identify them after the graph is released. 
For example, an adversary who knows the number of friends of all the victim 
vertices can later re-identify them in the released graph if every value 
happens to be unique, even if all vertex 
and edge attributes have been removed. Once a set of users is re-identified, 
the attacker can learn sensitive information, such as the existence 
of relations between two users or the (co-)affiliation of some of them 
to a community. 

According to the means by which adversaries obtain the knowledge used 
to re-identify the victims, they are classified as \emph{passive} or 
\emph{active} \cite{BDK2007}. Passive adversaries rely on information 
obtainable from publicly available sources, such as other OSNs, but do not 
attempt to purposely alter the structure of the network. On the other hand, 
active adversaries enroll sybil nodes in the network and try to force the 
creation of structural patterns that allow them to later re-identify 
the victims. Several active attacks were described 
in \cite{BDK2007,Peng2012,PLZW2014}. In every attack, the adversary 
inserts a small number of sybil nodes, and then creates unique connection 
patters, referred to as \emph{fingerprints}, between sybil nodes 
and the victims. Additionally, the connections between pairs of sybil nodes 
are established in such a way that the subgraph induced by them is not 
isomorphic (with high probability) to any other subgraph. Once the graph 
is published, the uniquely identifiable set of sybil nodes is retrived, 
and the unique fingerprints allow to re-identify the victims. 

The notion of $(k,\ell)$-anonymity was introduced in \cite{TY2016} as 
a measure of the resistance of a social graph to active attacks. 
Informally, a $(k,\ell)$-anonymous graph ensures that an adversary with 
the ability to insert up to $\ell$ sybil nodes in the network, cannot 
use the distances from these sybil nodes to other vertices to uniquely 
identify any vertex. This guarantee comes from the fact that each vertex 
is ensured to be undistinguishable from at least other $k-1$ vertices 
according to the so-called metric representation with respect to every 
vertex subset of size at most $\ell$. 
Then, a family of methods that transform a $(1,1)$-anonymous graph $G$ 
(which is the least private type of graphs) into a graph $G'$ 
that satsifies $(k,\ell)$-anonymity for $k>1$ or $\ell > 1$, was proposed 
in \cite{MauwTrujilloXuan2016,MauwRamirezTrujillo2016}. 

In this paper, we re-visit the notion of $(k,\ell)$-anonymity. We focus 
on two assumptions encoded in $(k,\ell)$-anonymity: treating every vertex 
subset of size up to $\ell$ as a potential set of sybil nodes, and assuming 
that the adversary is able to control the distances between the set of 
sybil nodes and every other vertex in the graph, which is not realistic. 
As a result, we first propose the notion of $(k,\ell)$-anonymous 
transformations, which ensure the same level of protection that 
would be achieved by enforcing $(k,\ell)$-anonymity while 
performing less modifications in the graph. Then, we introduce a new privacy 
property, $(k,\ell)$-adjacency anonymity, which relaxes the assumption on 
the distances that the adversary is able to control. Finally, these two ideas 
are combined in the notion of $(k,\ell)$-adjacency anonymous transformation,  
and we propose two methods, based on edge additions and removals, 
for performing $(k,1)$- and $(2,\ell)$-adjacency 
anonymous transformations. 

The remainder of this paper is structured as follows. 
In Section~\ref{sect-adversary-model} we discuss our 
new adversary model, introducing $(k,\ell)$-anonymous transformations, 
$(k,\ell)$-adjacency anonymity, and $(k,\ell)$-adjacency anonymous 
transformations. Section~\ref{sect-k-1-trans} introduces the algorithm 
for obtaining $(k,1)$-adjacency anonymous transformations, whereas 
Sect\-ion~\ref{sect-2-ell-trans} introduces the algorithm 
for obtaining $(2,\ell)$-adjacency anonymous transformations. Finally, 
we discuss our results and possible directions for future work in 
Section~\ref{sect-conclusions}. Before proceeding, we will introduce some 
notation that will be used throughout the paper. We will use the 
notation $u\sim_G v$ for two vertices $u,v\in V(G)$ that are adjacent 
in $G$, \emph{i.e.} $(u,v)\in E(G)$. The \emph{open neighbourhood} of a vertex 
$u\in V(G)$, denoted by $N_G(u)$, is the set $N_G(u)=\{v:\;u\sim_G v\}$, 
whereas the \emph{closed neighbourhood} of $u$ in $G$ is the set 
$N_G[u]=\{u\}\cup N_G(u)$. 
Similarly, for a set $S\subset V(G)$, we define 
$N_G(S)=\cup_{v\in S}N_G(v)\setminus S$ and $N_G[S]=\cup_{v\in S}N_G[v]$. 
The \emph{degree} of a vertex $u$, denoted by $\delta_G(u)$, is its number 
of neighbours, i.e. $\delta_G(u)=|N_G(u)|$. 
In a graph $G$ of order $n$, we will refer to vertices of degree 
$0$, $1$ and $n-1$ as \emph{isolated}, \emph{end-} and \emph{dominant} 
vertices, respectively. 
The \emph{distance} between two vertices $u$ and $v$ in $G$, 
denoted as $d_G(u,v)$, is the number of edges in a shortest path joining 
$u$ and $v$. 
For a graph $G=(V,E)$ and a subset $S$ 
of vertices of $G$, we will denote by $\langle S \rangle_G$ the subgraph 
of $G$ induced by $S$, that is $\langle S \rangle_G=(S,E\cap S\times S)$. 
In the previously defined notations, if there is no ambiguity, 
we will drop the graph-specific subindices and 
simply write $u\sim v$, $N(u)$, $\delta(u)$, etc. 
For a graph $G$, we define 
$\delta(G)=\min_{v\in V(G)}\{\delta_G(v)\}$ and 
$\Delta(G)=\max_{v\in V(G)}\{\delta_G(v)\}$ and, as usual, 
we will denote by $K_n$ and $N_n$ the complete and empty 
graphs of order $n$, respectively. 
%, that is, $K_n\cong(V,V\times V)$ and 
%$N_n\cong(V,\emptyset)$, where $|V|=n$. 

\section{Adversary model}\label{sect-adversary-model} 

An active adversary uses graph properties of a set of sybil 
nodes to re-identify users in an anonymised social graph. Prior publication of 
the social network graph, the active attacker adds a set 
of sybil 
nodes to the network (e.g. nodes $1,2,3$ and $4$ in Figure~\ref{fig-sybils}). 
The sybil nodes establish links between themselves and also with the victims 
(e.g. users $H$ and $G$ in Figure~\ref{fig-sybils}). 
After publication of the 
social network graph without the users' identifiers, the attacker
first searches for the subgraph formed by the sybil nodes. Victims
connected to the attacker subgraph can be reidentified by using the neighbour 
relation between sybil nodes and victims. For example, the non-sybil nodes
connected to $1$ and $4$ in Figure~\ref{fig-reidentification}, respectively, 
must be $H$ and $G$. 
This allows the 
adversary to acquire knowledge that was supposed to remain private, such as the 
existence of a link between users $H$ and $G$.

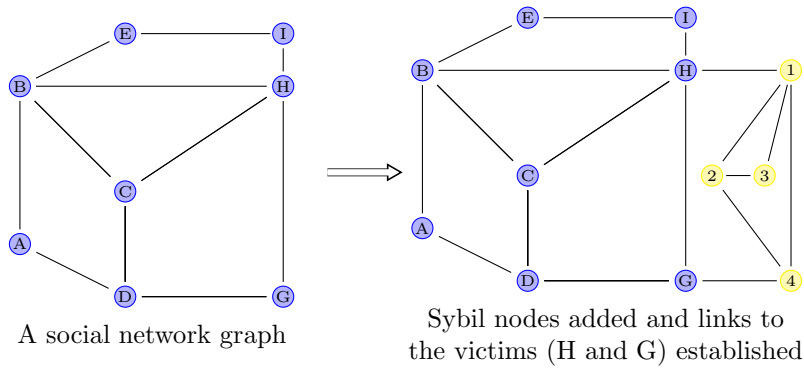
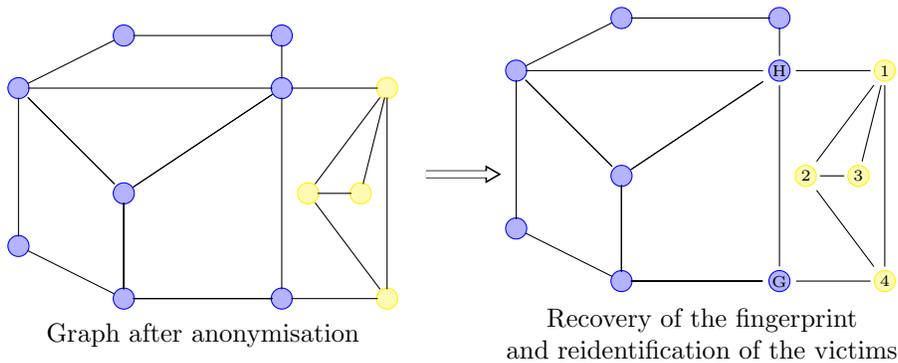
\begin{figure}%
\centering
\subfloat[Active attack prior publication.]{\label{fig-sybils}
\begin{minipage}{\columnwidth}\footnotesize
  \centering
  \subsubfloat{
	\begin{tikzpicture}[scale=0.7]

	\filldraw[fill opacity=0.3, blue, fill=blue]  (0,0) circle (0.2cm); 
	\node (C) at (0,0) {\tiny{C}};
	
	\filldraw[fill opacity=0.3, blue, fill=blue]  (-2,-1) circle (0.2cm); 
	\node (A) at (-2,-1) {\tiny{A}};
	
	\filldraw[fill opacity=0.3, blue, fill=blue]  (-2,2) circle (0.2cm); 
	\node (B) at (-2,2) {\tiny{B}};
	
	\filldraw[fill opacity=0.3, blue, fill=blue]  (0,-2) circle (0.2cm); 
	\node (D) at (0,-2) {\tiny{D}};
	
	\filldraw[fill opacity=0.3, blue, fill=blue]  (3,-2) circle (0.2cm); 
	\node (G) at (3,-2) {\tiny{G}};
	
	\filldraw[fill opacity=0.3, blue, fill=blue]  (0,3) circle (0.2cm); 
	\node (E) at (0,3) {\tiny{E}};
	
	\filldraw[fill opacity=0.3, blue, fill=blue]  (3,2) circle (0.2cm); 
	\node (H) at (3,2) {\tiny{H}};
	
	\filldraw[fill opacity=0.3, blue, fill=blue]  (3,3) circle (0.2cm); 
	\node (I) at (3,3) {\tiny{I}};
	
	\draw(A) -- (B) -- (C) -- (D) -- (G) -- (D) -- (C) -- (H);
	\draw(H) -- (C) -- (B) -- (H) -- (I) -- (E) -- (B);
	\draw(A) -- (D) -- (C);
	\draw (H) -- (G);
	
	\end{tikzpicture}
  }{A social network graph}%
  \qquad
  \hspace{-0.5cm}
  \subsubfloat{
	  \begin{tikzpicture}
	    \draw[	thick, decoration={markings,mark=at position
	       1 with {\arrow[semithick]{open triangle 60}}},
	       double distance=1.4pt, shorten >= 5.5pt,
	       preaction = {decorate},
	       postaction = {draw,line width=2.4pt, white,shorten >= 4.5pt}] (0,0) 
	       -- (1,0);
	  \end{tikzpicture}
  }
  \qquad
  \subsubfloat{
	\begin{tikzpicture}[scale=0.7]

	\filldraw[fill opacity=0.3, blue, fill=blue]  (0,0) circle (0.2cm); 
	\node (C) at (0,0) {\tiny{C}};
	
	\filldraw[fill opacity=0.3, blue, fill=blue]  (-2,-1) circle (0.2cm); 
	\node (A) at (-2,-1) {\tiny{A}};
	
	\filldraw[fill opacity=0.3, blue, fill=blue]  (-2,2) circle (0.2cm); 
	\node (B) at (-2,2) {\tiny{B}};
	
	\filldraw[fill opacity=0.3, blue, fill=blue]  (0,-2) circle (0.2cm); 
	\node (D) at (0,-2) {\tiny{D}};
	
	\filldraw[fill opacity=0.3, blue, fill=blue]  (3,-2) circle (0.2cm); 
	\node (G) at (3,-2) {\tiny{G}};
	
	\filldraw[fill opacity=0.3, blue, fill=blue]  (0,3) circle (0.2cm); 
	\node (E) at (0,3) {\tiny{E}};
	
	\filldraw[fill opacity=0.3, blue, fill=blue]  (3,2) circle (0.2cm); 
	\node (H) at (3,2) {\tiny{H}};
	
	\filldraw[fill opacity=0.3, blue, fill=blue]  (3,3) circle (0.2cm); 
	\node (I) at (3,3) {\tiny{I}};
	
	\draw(A) -- (B) -- (C) -- (D) -- (G) -- (D) -- (C) -- (H);
	\draw(H) -- (C) -- (B) -- (H) -- (I) -- (E) -- (B);
	\draw(A) -- (D) -- (C);
	\draw (H) -- (G);
	
	\filldraw[fill opacity=0.3, yellow, fill=yellow]  (3.5,0) circle 
	(0.2cm); 
	\node (a3) at (3.5,0) {\tiny{$2$}};
	
	\filldraw[fill opacity=0.3, yellow, fill=yellow]  (5,2) circle (0.2cm); 
	\node (a2) at (5,2) {\tiny{$1$}};
	
	\filldraw[fill opacity=0.3, yellow, fill=yellow]  (4.5,0) circle (0.2cm); 
	\node (a4) at (4.5,0) {\tiny{$3$}};
	
	\filldraw[fill opacity=0.3, yellow, fill=yellow]  (5,-2) circle (0.2cm); 
	\node (a5) at (5,-2) {\tiny{$4$}};
	
	\draw(a3) -- (a2) -- (a4) -- (a3) -- (a5);
	
	\draw(a2) -- (a5);
	
	\draw(a2) -- (H);
	
	\draw(a5) -- (G);
	
	\end{tikzpicture}
  
  }{Sybil nodes added and links to \\ the victims (H and G) established}
    \\  \medskip
\end{minipage}
}

\subfloat[Reidentification after publication.]{\label{fig-reidentification}
\begin{minipage}{\columnwidth}\footnotesize
  \centering

  \subsubfloat{
	\begin{tikzpicture}[scale=0.7]

	\filldraw[fill opacity=0.3, blue, fill=blue]  (0,0) circle (0.2cm); 
	\node (C) at (0,0) {\tiny{}};
	
	\filldraw[fill opacity=0.3, blue, fill=blue]  (-2,-1) circle (0.2cm); 
	\node (A) at (-2,-1) {\tiny{}};
	
	\filldraw[fill opacity=0.3, blue, fill=blue]  (-2,2) circle (0.2cm); 
	\node (B) at (-2,2) {\tiny{}};
	
	\filldraw[fill opacity=0.3, blue, fill=blue]  (0,-2) circle (0.2cm); 
	\node (D) at (0,-2) {\tiny{}};
	
	\filldraw[fill opacity=0.3, blue, fill=blue]  (3,-2) circle (0.2cm); 
	\node (G) at (3,-2) {\tiny{}};
	
	\filldraw[fill opacity=0.3, blue, fill=blue]  (0,3) circle (0.2cm); 
	\node (E) at (0,3) {\tiny{}};
	
	\filldraw[fill opacity=0.3, blue, fill=blue]  (3,2) circle (0.2cm); 
	\node (H) at (3,2) {\tiny{}};
	
	\filldraw[fill opacity=0.3, blue, fill=blue]  (3,3) circle (0.2cm); 
	\node (I) at (3,3) {\tiny{}};
	
	\draw(A) -- (B) -- (C) -- (D) -- (G) -- (D) -- (C) -- (H);
	\draw(H) -- (C) -- (B) -- (H) -- (I) -- (E) -- (B);
	\draw(A) -- (D) -- (C);
	\draw (H) -- (G);
	
	\filldraw[fill opacity=0.3, yellow, fill=yellow]  (3.5,0) circle 
	(0.2cm); 
	\node (a3) at (3.5,0) {\tiny{}};
	
	\filldraw[fill opacity=0.3, yellow, fill=yellow]  (5,2) circle (0.2cm); 
	\node (a2) at (5,2) {\tiny{}};
	
	\filldraw[fill opacity=0.3, yellow, fill=yellow]  (4.5,0) circle (0.2cm); 
	\node (a4) at (4.5,0) {\tiny{}};
	
	\filldraw[fill opacity=0.3, yellow, fill=yellow]  (5,-2) circle (0.2cm); 
	\node (a5) at (5,-2) {\tiny{}};
	
	\draw(a3) -- (a2) -- (a4) -- (a3) -- (a5);
	
	\draw(a2) -- (a5);
	
	\draw(a2) -- (H);
	
	\draw(a5) -- (G);
	
	\end{tikzpicture}
  }{Graph after anonymisation}%
  \qquad
  \hspace{-0.5cm}
  \subsubfloat{
	  \begin{tikzpicture}
	    \draw[	thick, decoration={markings,mark=at position
	       1 with {\arrow[semithick]{open triangle 60}}},
	       double distance=1.4pt, shorten >= 5.5pt,
	       preaction = {decorate},
	       postaction = {draw,line width=2.4pt, white,shorten >= 4.5pt}] (0,0) 
	       -- (1,0);
	  \end{tikzpicture}
  }
  \qquad
  \subsubfloat{
	\begin{tikzpicture}[scale=0.7]

	\filldraw[fill opacity=0.3, blue, fill=blue]  (0,0) circle (0.2cm); 
	\node (C) at (0,0) {\tiny{}};
	
	\filldraw[fill opacity=0.3, blue, fill=blue]  (-2,-1) circle (0.2cm); 
	\node (A) at (-2,-1) {\tiny{}};
	
	\filldraw[fill opacity=0.3, blue, fill=blue]  (-2,2) circle (0.2cm); 
	\node (B) at (-2,2) {\tiny{}};
	
	\filldraw[fill opacity=0.3, blue, fill=blue]  (0,-2) circle (0.2cm); 
	\node (D) at (0,-2) {\tiny{}};
	
	\filldraw[fill opacity=0.3, blue, fill=blue]  (3,-2) circle (0.2cm); 
	\node (G) at (3,-2) {\tiny{G}};
	
	\filldraw[fill opacity=0.3, blue, fill=blue]  (0,3) circle (0.2cm); 
	\node (E) at (0,3) {\tiny{}};
	
	\filldraw[fill opacity=0.3, blue, fill=blue]  (3,2) circle (0.2cm); 
	\node (H) at (3,2) {\tiny{H}};
	
	\filldraw[fill opacity=0.3, blue, fill=blue]  (3,3) circle (0.2cm); 
	\node (I) at (3,3) {\tiny{}};
	
	\draw(A) -- (B) -- (C) -- (D) -- (G) -- (D) -- (C) -- (H);
	\draw(H) -- (C) -- (B) -- (H) -- (I) -- (E) -- (B);
	\draw(A) -- (D) -- (C);
	\draw (H) -- (G);
	
	\filldraw[fill opacity=0.3, yellow, fill=yellow]  (3.5,0) circle 
	(0.2cm); 
	\node (a3) at (3.5,0) {\tiny{$2$}};
	
	\filldraw[fill opacity=0.3, yellow, fill=yellow]  (5,2) circle (0.2cm); 
	\node (a2) at (5,2) {\tiny{$1$}};
	
	\filldraw[fill opacity=0.3, yellow, fill=yellow]  (4.5,0) circle (0.2cm); 
	\node (a4) at (4.5,0) {\tiny{$3$}};
	
	\filldraw[fill opacity=0.3, yellow, fill=yellow]  (5,-2) circle (0.2cm); 
	\node (a5) at (5,-2) {\tiny{$4$}};
	
	\draw(a3) -- (a2) -- (a4) -- (a3) -- (a5);
	
	\draw(a2) -- (a5);

	\draw(a2) -- (H);

	\draw(a5) -- (G);
	
	\end{tikzpicture}
  
  }{Recovery of the fingerprint \\ and reidentification of the victims}
  \medskip
\end{minipage}
}

\caption{The four stages of an active attack. \label{fig-active-attack}}
\end{figure}

From a practical point of view, active attacks require the ability to insert 
sybil nodes in the social network and remain unnoticed by 
sybil detection techniques. This is a fairly easy task in today's social 
networks, as false 
positives in sybil detection are undesirable and registration to 
the network should be trivial; social networks understandably 
favour usability and user experience over sybil detection. 
From a theoretical point of view, an active attack relies on creating a 
\emph{unique} attacker subgraph. That is to say, the induced subgraph formed by 
the sybil nodes should have no trivial automorphism and no other subgraph 
in the network isomorphic to it. For example, assume that the adversary in the 
attack in Figure~\ref{fig-active-attack} could not insert the third node, i.e. 
the node labelled 
$3$. This makes the subgraph induced by $1, 2$ and $4$ isomorphic to the 
subgraph induced by $B$, $C$, and $H$, which prevents the attacker 
from correctly retrieving the inserted subgraph. 

Backstrom et al. already showed that, despite of the 
previously mentioned challenges, active attacks can be implemented 
successfully~\cite{BDK2007}. They proved that only $\log n$ sybil 
nodes, where $n$ is the number of vertices of the network, 
are needed to create an attacker subgraph which is unique 
with high probability. That makes active attacks particularly 
dangerous and hard to prevent. 

Effectively determining whether a social graph is vulnerable to an active 
attack 
is a necessary step towards developing a mitigation strategy against it. For 
example, the complete graph $G$ satisfies that for every proper subgraph $S$ 
there exists another subgraph $S'$ that is isomorphic to $S$. 
Such property 
makes an active attack unfeasible in a complete graph.  Determining the actual 
resistance of an arbitrary graph to active attacks is not 
trivial, though.
A first step on this direction was given in~\cite{TY2016}, where 
Trujillo-Rasua and Yero introduced the privacy measure $(k, \ell)$-anonymity. 

Consider a total order $\preceq$ on the vertices of a graph $G$. 
Given a set $S\subseteq V(G)$, let $(v_1,v_2,\ldots,v_t)$, where $v_i\in S$ 
for $i\in\{1,\ldots,t\}$, be the vector  
composed by the elements of $S$, in such a way that 
$v_1 \preceq v_2 \preceq \ldots \preceq v_t$. 
For the sake of simplicity in our presentation, 
in what follows we will abuse notation and refer to the 
\emph{ordered set} $S=\{v_1,v_2,\ldots,v_t\}$. 
%The \emph{metric representation} of a vertex $v\in V(G)$ with respect 
%to an ordered set $S=\{v_1,v_2,\ldots,v_t\}\subseteq V(G)$ is 
%the vector $r_G(v\;|\;S)=(d(v,v_1), d(v,v_2), \ldots, d(v,v_t))$. 
Given an ordered set of sybil nodes $S = (s_1, \ldots, s_t)$ in a graph 
$G = (V, E)$, Trujillo-Rasua and Yero~\cite{TY2016} define the adversary 
knowledge about a user $u \in V$ as the vector 
$(d_G(v, s_1), \ldots, d_G(v, s_t))$. This vector is referred to as 
\emph{metric representation} of $u$ with respect to $S$, and denoted 
$r_G(u | S)$ \cite{Slater1975,Harary1976}. 

The metric representation was introduced in \cite{Slater1975,Harary1976} 
as a tool to define the so-called \emph{resolving sets}. 
A set $S\subset V(G)$ is said to be a 
resolving set\footnote{Multiple terminologies have been used for 
resolving sets in the literature. The term \emph{resolving set} was introduced 
by Slater in \cite{Slater1975}, whereas Harary and Melter \cite{Harary1976} 
introduced the terminology \emph{locating set} and Seb\H{o} and Tannier 
\cite{Sebo2004} preferred the term \emph{metric generator}.} 
of $G$ if every vertex $u\in V(G)$ has a unique metric representation 
with respect to $S$. 
This property of resolving sets inspired the definition in 
\cite{TY2016} of an opposed concept, with implications in vertex privacy. 

\begin{definition}[$k$-antiresolving set]\label{def-antires} 
Let $G=(V,E)$ be a non-trivial graph. A set $S\subset V$ 
is a \emph{$k$-antiresolving set} of $G$ if $k$ is 
the largest positive integer such that, 
for every $v\in V(G)\setminus S$, there exist vertices 
$w_1,w_2,\ldots,w_{k-1}\in V(G)\setminus S$ such 
that $v,w_1,w_2,\ldots,w_{k-1}$ are pairwise different and 
$$r_G(v\;|\;S)=r_G(w_1\;|\;S)=r_G(w_2\;|\;S)=\ldots=r_G(w_{k-1}\;|\;S).$$
\end{definition}

The minimum cardinality 
of a $k$-antiresolving set of a graph $G$ is called the $k$-metric 
antidimension of $G$. These concepts were used to quantify the privacy 
of a social graph in the presence of active attackers as follows. 

\begin{definition}[$(k,\ell)$-anonymity]\label{def-k-ell-anonymity}
A graph $G$ is said to satisfy \emph{$(k,\ell)$-anonymity} if $k$ 
is the smallest positive integer such that the $k$-metric antidimension 
of $G$ is smaller than or equal to $\ell$. 
\end{definition}

From a privacy perspective, if a graph satisfies $(k, \ell)$-anonymity, 
an attacker with the capacity to enrol, and successfully retrieve, 
up to $\ell$ sybil nodes in the graph would still be incapable 
of distinguishing any vertex from at least other $k-1$ vertices. 
Taking back again the example of the complete graph $K_n$, it is easy to prove 
that $K_n$ satisfies $(n-l, l)$-anonymity. It is worth noticing that $k = n-l$ 
corresponds to the
maximum value possible for $k$ in $(k, \ell)$-anonymity given $\ell = l$.  

Certainly, a graph satisfying $(k, \ell)$-anonymity for $k>1$ effectively 
resists active attacks when performed by at most $\ell$ sybil nodes. However, 
event the simplest of the privacy goals, namely transforming a 
$(1,1)$-anonymous graph into a $(k, \ell)$-anonymous graph $G'$ with either $k 
> 1$ or $\ell > 1$, has not been accomplished without significant information 
loss~\cite{MauwTrujilloXuan2016}. Our observation is that $(k, 
\ell)$-anonymity, although suitable to quantify resistance against active 
attacks, cannot be applied straightforwardly to privacy-preserving 
transformation of social graphs. 

\subsection{Revisiting $(k, \ell)$-anonymity}

$(k, 
\ell)$-anonymity quantifies over all subsets of vertices of size at most 
$\ell$. Therefore, a transformation from an original graph $G$ to an anonymized 
graph $G'$ satisfying, for example $(2, 
\ell)$-anonymity, must ensure that every subset of vertices $S$ in $G'$ with 
$|S| \leq \ell$ is a $k'$-antiresolving set where $k' \geq 2$, regardless of 
whether $S$ was indeed a $1$-antiresolving set in $G$ or not. In effect, 
assuming that the set of attacker nodes $S$ is already a $2$-antiresolving set 
in the original graph $G$, it is harmless to publish (with respect to $(2, 
|S|)$-anonymity ) a transformation of 
$G$ where $S$ is a $1$-antiresolving set. Consequently, when aiming at $(2, 
\ell)$-anonymity, a transformation method should only be concerned about those 
$1$-antiresolving sets with size at most $\ell$. We formalise this concept as 
follows.

\begin{definition}[$(k, \ell)$-anonymous transformation]
A pair $(G_1, G_2)$ is a \emph{$(k, \ell)$-ano\-nym\-ous transformation} if for 
every 
subset $S \subseteq V(G_1) \cap V(G_2)$ with $|S| \leq \ell$, $S$ being a 
$k_1$-antiresolving set in $G_1$, and $S$ being $k_2$-antiresolving set in 
$G_2$, it holds that $k_1 < k \implies k_2 \geq k$.
\end{definition}

Notice that, in particular, if a graph $G$ satisfies $(k,\ell)$-anonymity, 
then every pair $(G_0,G)$, where $G_0$ is an arbitrary graph, 
is a $(k,\ell)$-anonymous transformation. The converse is not true, 
as exemplified in Figure~\ref{fig-ex-k-ell-transform}. 

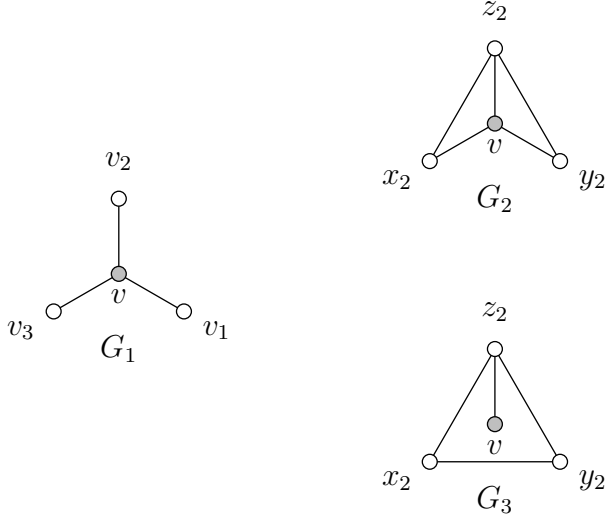
\begin{figure}[!ht]
\begin{center}

\begin{tikzpicture}[inner sep=0.7mm, place/.style={circle,draw=black,
fill=white},xx/.style={circle,draw=black!99, 
fill=black!99},gray1/.style={circle,draw=black!99, 
fill=black!25},gray2/.style={circle,draw=black!99, 
fill=black!50},gray3/.style={circle,draw=black!99,fill=black!75},
transition/.style={rectangle,draw=black!50,fill=black!20,thick}, 
line width=.5pt]

\def\radius{1cm}
\def\lblradius{1.5cm}

\def\n{3}

\foreach \ind in {1,...,3}\pgfmathparse{90+360/\n*\ind}\coordinate (g2\ind) at 
(\pgfmathresult:\radius);

\foreach \ind in {1,...,3}\pgfmathparse{90+360/\n*\ind}\coordinate (v2\ind) at (\pgfmathresult:\lblradius);

\foreach \ind in {1,...,3}\pgfmathparse{90+360/\n*\ind}\coordinate[xshift=5cm,yshift=2cm] (g1\ind) at (\pgfmathresult:\radius);

\foreach \ind in {1,...,3}\pgfmathparse{90+360/\n*\ind}\coordinate[xshift=5cm,yshift=2cm] (v1\ind) at (\pgfmathresult:\lblradius); 

\foreach \ind in {1,...,3}\pgfmathparse{90+360/\n*\ind}\coordinate[xshift=5cm,yshift=-2cm] (g3\ind) at (\pgfmathresult:\radius);

\foreach \ind in {1,...,3}\pgfmathparse{90+360/\n*\ind}\coordinate[xshift=5cm,yshift=-2cm] (v3\ind) at (\pgfmathresult:\lblradius); 

\coordinate (cv1) at (0,0);
\coordinate (cv2) at (5,2);
\coordinate (cv3) at (5,-2);

\draw[black] (g21) -- (cv1) -- (g22);
\draw[black] (cv1) -- (g23);

\draw[black] (g13) -- (cv2) -- (g11) -- (g13) -- (g12) -- (cv2);

\draw[black] (cv3) -- (g33) -- (g32) -- (g31) -- (g33);

\node [place] at (g21) {};
\node at (v21) {$v_3$};
\node [place] at (g22) {};
\node at (v22) {$v_1$};
\node [place] at (g23) {};
\node at (v23) {$v_2$};

\node [place] at (g11) {};
\node at (v11) {$x_2$};
\node [place] at (g12) {};
\node at (v12) {$y_2$};
\node [place] at (g13) {};
\node at (v13) {$z_2$};

\node [place] at (g31) {};
\node at (v31) {$x_2$};
\node [place] at (g32) {};
\node at (v32) {$y_2$};
\node [place] at (g33) {};
\node at (v33) {$z_2$};

\node [gray1] at (cv1) {};
\coordinate [label=center:{$v$}] (cv1l) at (0,-0.3);
\node [gray1] at (cv2) {};
\coordinate [label=center:{$v$}] (cv2l) at (5,1.7);
\node [gray1] at (cv3) {};
\coordinate [label=center:{$v$}] (cv2l) at (5,-2.3);

\coordinate [label=center:{$G_1$}] (G1) at (0,-1);
\coordinate [label=center:{$G_2$}] (G2) at (5,1);
\coordinate [label=center:{$G_3$}] (G3) at (5,-3);

\end{tikzpicture}

\end{center}
\caption{Two $(2,1)$-anonymous transformations $(G_1,G_2)$ and $(G_1,G_3)$. 
The graph $G_2$ satisfies $(2,1)$-anonymity, whereas $G_3$ does not.} 
\label{fig-ex-k-ell-transform}
\end{figure} 

\subsection{The adversary knowledge}

Privacy measures based on $k$-anonymity are defined 
based on a concrete definition of the \emph{adversary knowledge}. 
In $(k, \ell)$-anonymity, an adversary is a set of sybil nodes 
$S \subseteq V(G)$ within a network $G$. The knowledge of such adversary 
about a user $u \in V(G)-S$ is considered to be the metric representation 
$r_G(u | S)$. That is to say, the adversary is capable of determining 
the distance from every attacker node to any other node in the network. 
This is a strong assumption, yet it can be justified by the necessity 
of not underestimating the adversary capabilities. 

In this article we relax the assumption on the adversary knowledge made 
in~\cite{TY2016}. Our decision is based on the fact that all active attacks 
proposed so far~\cite{BDK2007,Peng2012,PLZW2014} rely on the neighbour 
relation between the attacker nodes and the victims. It is indeed 
unrealistic to expect the adversary to rely on arbitrary distances, 
since that would imply knowing the entire adjacency matrix and, especially, 
having the capability to influence whether a relation is established, or not, 
between any pair of users of the network. 

In a manner analogous to the definition of antiresolving sets, 
we use standard concepts from Graph theory to represent an adversary 
that only has knowledge about its neighbours. The concept is known as 
\emph{adjacency representation}, introduced by Jannesari and Omoomi 
\cite{JanOmo2012} and defined as follows. 

\begin{definition}[Adjacency representation]
Given a graph $G=(V,E)$, an ordered set $S = \{s_1, \ldots, s_t\} \subset V$, 
and a vertex $u \in V$, 
the \emph{adjacency metric representation} of $u$ with respect $S$ is the 
vector $a_G(r | S ) = (a_G(s_1,v), \ldots, a_G(s_t,v))$ where $a_G: V(G) \times 
V(G) \rightarrow \mathbb{N}$ is 
defined by:

\begin{equation}
a_G(u,v)=\left\{\begin{array}{rl}
0&\textrm{ if  }u=v\\
1&\textrm{ if  }u\sim_{_G} v\\
2&\textrm{ otherwise}
\end{array}\right.
\end{equation}
\end{definition}

Note that $a_G(u,v)=\min\{2,d_G(u,v)\}$ for every $u,v\in V(G)$. Now, we 
will adapt the notion of $k$-antiresolving sets in order to account for the 
new type of adversary. 

\begin{definition}[$k$-adjacency antiresolving set] 
\label{def-k-adj-antiresolving-set} 
Let $G=(V,E)$ be a non-trivial graph. A set $S\subset V$ is a 
\emph{$k$-adjacency antiresolving set} of $G$ if $k$ is the largest positive 
integer 
such that, for every $v\in V(G)\setminus S$, 
there exist vertices $w_1,w_2,\ldots,w_{k-1}$ such 
that $v,w_1,w_2,\ldots,w_{k-1}$ are pairwise different and 
$$a_G(v\;|\;S)=a_G(w_1\;|\;S)=a_G(w_2\;|\;S)=\ldots=a_G(w_{k-1}\;|\;S).$$
\end{definition}

To illustrate the difference between $k$-adjacency antiresolving sets and 
$k$-antiresolv\-ing sets, consider the graph $G$ depicted 
in Figure~\ref{fig-example1}. 
The set $\{v\}$ is a $2$-antiresolving set of $G$, as $z_1$ and $z_2$ satisfy  
$d(v,z_1)=d(v,z_2)=3$, whereas $d(v,x_1)=d(v,x_2)=d(v,x_3)=d(v,x_4)=1$ and 
$d(v,y_1)=d(v,y_2)=d(v,y_3)=d(v,y_4)=2$. On the other hand, we have that 
$a_{G}(x_1,\{v\})=a_{G}(x_2,\{v\})=a_{G}(x_3,\{v\})=a_{G}(x_4,\{v\})=(1)$, 
while $a_{G}(y_1,\{v\})=a_{G}(y_2,\{v\})=a_{G}(y_3,\{v\})=a_{G}(y_4,\{v\})= 
a_{G}(z_1,\{v\})=a_{G}(z_2,\{v\})=(2)$, so $\{v\}$ is a 
$4$-adjacency 
antiresolving set of $G$.

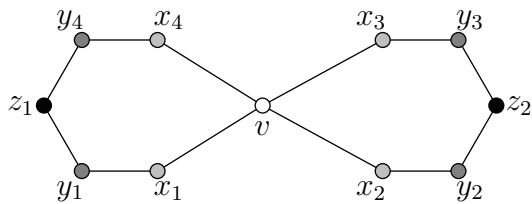
\begin{figure}[!ht]
\begin{center}
\begin{tikzpicture}[inner sep=0.7mm, place/.style={circle,draw=black,
fill=white},xx/.style={circle,draw=black!99, 
fill=black!99},gray1/.style={circle,draw=black!99, 
fill=black!25},gray2/.style={circle,draw=black!99, 
fill=black!50},gray3/.style={circle,draw=black!99,fill=black!75},
transition/.style={rectangle,draw=black!50,fill=black!20,thick}, 
line width=.5pt,scale=0.38]

\def\radius{2.65cm}
\def\lblradius{3.45cm}

\def\n{6}

\foreach \ind in {1,...,5}\pgfmathparse{360/\n*\ind}\coordinate (g2\ind) at 
(\pgfmathresult:\radius);

\foreach \ind in {1,...,5}\pgfmathparse{360/\n*\ind}\coordinate (v2\ind) at (\pgfmathresult:\lblradius);

\draw[black] (g21) -- (g22) -- (g23) -- (g24) -- (g25);

\foreach \ind in {1,...,5}\pgfmathparse{180+360/\n*\ind}\coordinate[xshift=4cm] (g1\ind) at (\pgfmathresult:\radius);

\foreach \ind in {1,...,5}\pgfmathparse{180+360/\n*\ind}\coordinate[xshift=4cm] (v1\ind) at (\pgfmathresult:\lblradius); 

\draw[black] (g11) -- (g12) -- (g13) -- (g14) -- (g15);

\coordinate (cv) at (5,0);

\draw[black] (g21) -- (cv) -- (g25);
\draw[black] (g15) -- (cv) -- (g11);

\node [gray1] at (g21) {};
\node at (v21) {$x_4$};
\node [gray2] at (g22) {};
\node at (v22) {$y_4$};
\node [xx] at (g23) {};
\node at (v23) {$z_1$};
\node [gray2] at (g24) {};
\node at (v24) {$y_1$};
\node [gray1] at (g25) {};
\node at (v25) {$x_1$};

\node [gray1] at (g11) {};
\node at (v11) {$x_2$};
\node [gray2] at (g12) {};
\node at (v12) {$y_2$};
\node [xx] at (g13) {};
\node at (v13) {$z_2$};
\node [gray2] at (g14) {};
\node at (v14) {$y_3$};
\node [gray1] at (g15) {};
\node at (v15) {$x_3$};

\node [place] at (cv) {};
\coordinate [label=center:{$v$}] (cvl) at (5,-0.8);

\end{tikzpicture}

\end{center}
\caption{In this graph, the set $\{v\}$ is a $2$-antiresolving set 
and a $4$-adjacency antiresolving set.}
\label{fig-example1}
\end{figure}

%\begin{definition}\label{def-equiv-rel-same-adj-repr}
For a graph $G=(V,E)$ and a set $S\subset V$, let $R_{G,S}$ be the equivalence 
relation such that two vertices $u$ and $v$ satisfy $u\ R_{G,S}\ v$ 
if and only if $u,v\in V\setminus S$ and $a_G(u\;|\;S)=a_G(v\;|\;S)$. 
%\end{definition}
%For a graph $G=(V,E)$ and a set $S\subset V$, 
Moreover, 
we will use the notation  
$\mathcal{A}_{G,S}$ for the set of equivalence 
classes induced in $V\setminus S$ by the relation $R_{G,S}$. 
It is simple to 
see that $S$ is a $(\min_{A\in \mathcal{A}_{G,S}}\{|A|\})$-adjacency 
antiresolving set of $G$.

\subsection{Problem statement}

We will first enunciate the notions of $(k,\ell)$-adjacency anonymity 
and $(k,\ell)$-adjacency anonymous transformation, which restrict 
the original definitions of $(k,\ell)$-anonymity and $(k,\ell)$-anonymous 
transformation, to account for adversaries whose knowledge consists 
of the adjacency representations of their victims. 

\begin{definition}[$k$-adjacency antidimension]\label{def-k-adj-antidim} 
The $k$-adjacency antidimension of a graph $G$ is the minimum cardinality of 
a $k$-adjacency antiresolving set of $G$.
\end{definition} 

\begin{definition}[$(k,\ell)$-adjacency anonymity]\label{def-k-ell-adj-anon} 
A graph $G$ satisfies $(k,\ell)$-adjacency ano\-nym\-ity if $k$ is the smallest 
positive integer such that the $k$-adjacency antidimension of $G$ is smaller 
than or equal to $\ell$.
\end{definition}

According to Definition~\ref{def-k-ell-adj-anon}, if a graph $G$ satisfies 
$(k,\ell)$-adjacency anonymity, then for every $S\subset V(G)$ of size 
at most $\ell$ and every $u\in V(G)\setminus S$ there 
exist $v_1,\ldots,v_{k-1}\in V(G)\setminus(S\cup\{u\})$ such that 
$u,v_1,\ldots,v_{k-1}$ are pairwise different and 
$a_G(u\;|\;S)=a_G(v_1\;|\;S)=\ldots=a_G(v_{k-1}\;|\;S)$, so the probability 
of $S$ being able to re-identify $v$ is at most $1/k$. 

It is simple to see that the complete graph $K_n$ and the empty graph 
$N_n=(V,\emptyset)$ satisfy $(n-\ell,\ell)$-adjacency anonymity 
for every $\ell\in\{1,n-1\}$, because for every $S\subset V$ the 
adjacency representation of every other vertex with respect to $S$ is 
either $(1,1,\ldots,1)$ or $(2,2,\ldots,2)$, respectively. In the next  
sections we will introduce results characterising the graphs that satisfy 
$(k,\ell)$-adjacency anonymity for other values of $k$ and $\ell$. 

\begin{definition}[$(k,\ell)$-adjacency anonymous transformation] 
A pair $(G_1, G_2)$ is a \emph{$(k, \ell)$-adjacency anonymous transformation} 
if for every subset $S \subseteq V(G_1) \cap V(G_2)$ 
with $|S| \leq \ell$, $S$ being a $k_1$-adjacency antiresolving set in $G_1$, 
and $S$ being $k_2$-adjacency antiresolving set in $G_2$, it holds that 
$k_1 < k \implies k_2 \geq k$. 
\end{definition}

In a manner analogous to $(k,\ell)$-anonymous transformations, we have 
that if a graph $G$ satisfies $(k,\ell)$-adjacency anonymity, 
then every pair $(G_0,G)$, where $G_0$ is an arbitrary graph, 
is a $(k,\ell)$-adjacency anonymous transformation. 
%The converse is not true, as exemplified 
%in Figure~\ref{fig-ex-k-ell-transform}. 

\begin{definition}[\textbf{Problem statement}]
Let $\loss(G, G')$ be a cost function providing the information loss incurred 
by the graph transformation from $G$ to $G'$. Given a graph $G$, 
and natural numbers $k$ and $\ell$, find $G'$ such that $(G, G')$ 
is a $(k,\ell)$-adjacency anonymous transformation and $\loss(G, G')$ 
is minimum. 
\end{definition} 

\section{$(k,1)$-adjacency anonymous transformations} 
\label{sect-k-1-trans} 

%As we mentioned previously, an edge-addition method for transforming a 
%$(1,1)$-anonymous 
%graph $G$ into a graph $G'$ that satsifies $(k,\ell)$-anonymity 
%for $k>1$ or $\ell > 1$ was proposed in  
%\cite{MauwTrujilloXuan2016,MauwRamirezTrujillo2016}. This method requires 
%to frequently compute and/or update the distance matrix of $G$ 
%as it is modified, which is considerably time consuming. 

Consider a $(k_0,1)$-adjacency anonymous graph $G$ of order $n$. In order 
to increase the resistance of $G$ to active attackers leveraging one sybil 
node, our interest is to propose $(k,1)$-adjacency transformations 
of the form $(G,G')$ where $k>k_0$. The next result allows us to assess 
the values of $k$ that may be of interest. 

\begin{proposition}\label{prop-too-large-k}
Let $G$ be a non-complete, non-empty graph of order $n$ satisfying 
$(k,1)$-adjacency anonymity. Then, $k\le\left\lfloor\frac{n-1}{2}\right\rfloor$. 
\end{proposition}

\begin{proof}
Let $G=(V,E)$ be a non-complete, non-empty graph of order $n$ satisfying 
$(k,1)$-adjacency anonymity. Suppose, for the purpose of contradiction, 
that $k>\left\lfloor\frac{n-1}{2}\right\rfloor$. Let $v\in V$ be a vertex 
of $G$ satisfying $\delta(v)\notin\{0,n-1\}$. The existence of such a vertex 
is guaranteed by the fact that the graph is not complete nor empty. 
We have that $\mathcal{A}_{G,\{v\}}=\{N_G(v), V\setminus N_G[v]\}$. 
If $\delta(v)\le\left\lfloor\frac{n-1}{2}\right\rfloor$, then 
$\{v\}$ is a $k'$-adjacency antiresolving set of $G$ with $k'<k$, which is 
a contradiction. On the other hand, 
if $\delta(v)>\left\lfloor\frac{n-1}{2}\right\rfloor$, 
then $|V\setminus N_G[v]|\le\left\lfloor\frac{n-1}{2}\right\rfloor$, which 
again means that $\{v\}$ is a $k'$-adjacency antiresolving set of $G$ 
with $k'<k$, a contradiction. Therefore, we have that 
$k\le\left\lfloor\frac{n-1}{2}\right\rfloor$.
\end{proof}

According to Proposition~\ref{prop-too-large-k}, in order to enforce 
$(k,1)$-adjacency anonymity on $G$ for some 
$k>\left\lfloor\frac{n-1}{2}\right\rfloor$, it is necessary to 
transform $G$ into a complete or empty graph, which lacks interest for us 
because such a graph would be completely useless for analysis. Thus, we 
will focus on the values of $k$ in the interval 
$\left[k_0+1,\left\lfloor\frac{n-1}{2}\right\rfloor\right]$.

The following results show the relations between the minimum and maximum 
degrees of a graph and its resistance against active attackers leveraging 
one sybil node. We first introduce some additional notation. 
For a graph $G=(V,E)$, let $I_G=\{v\in V:\;\delta(v)=0\}$ be the set of 
isolated vertices and let $D_G=\{v\in V:\;\delta(v)=n-1\}$ be the set of 
dominant vertices. Clearly, either $I_G=D_G=\emptyset$, or 
$I_G=\emptyset \wedge D_G\neq \emptyset$, or $I_G\neq\emptyset \wedge 
D_G=\emptyset$. With these definitions in mind, we give the 
following three results. 

\begin{proposition}\label{prop-k1-anonymity-noiso-nodom} 
Every non-complete graph $G$ $n$ such that $I_G=\emptyset$ satisfies 
$(k,1)$-adjacency anonymity 
with $k=\min\left\{\delta(G), n-\Delta(G)-1\right\}$. 
\end{proposition} 

\begin{proof}
Let $G=(V,E)$ be a graph of order $n$ without isolated or dominant vertices. 
Consider a vertex $v\in V$. Clearly, 
$\mathcal{A}_{G,\{v\}}=\{N_G(v), V\setminus N_G[v]\}$, 
so the set $\{v\}$ is a $(\min\{\delta(v),n-\delta(v)-1\})$-adjacency 
antiresolving set of $G$. In consequence, the graph $G$ satisfies 
$(k,1)$-adjacency anonymity with  
\begin{displaymath} 
\begin{array}{rcl} 
k&=&\displaystyle\min_{v\in V}\left\{\min\left\{\delta(v), 
n-\delta(v)-1\right\}\right\}\\ 
&=&\displaystyle\min\left\{\min_{v\in V}\left\{\delta(v), 
n-\delta(v)-1\right\}\right\}\\ 
&=&\displaystyle\min\left\{\min_{v\in V}\left\{\delta(v)\right\}, 
\min_{v\in V}\left\{n-\delta(v)-1\right\}\right\}\\ 
&=&\displaystyle\min\left\{\delta(G), n-\Delta(G)-1\right\}
\end{array} 
\end{displaymath} 
\end{proof} 

\begin{proposition}\label{prop-k1-anonymity-noiso-yesdom}
Let $G=(V,E)$ be a non-complete graph of order $n$ such that 
$D_G\neq \emptyset$ 
and let $S=V\setminus D_G$. Then, $G$ satisfies $(k,1)$-adjacency 
anonymity with $$k=\min\left\{\delta(G), 
|S|-\Delta(\langle S \rangle_G)-1\right\}.$$
\end{proposition}

\begin{proof}
We follow a reasoning analogous to that of the proof of 
Proposition~\ref{prop-k1-anonymity-noiso-nodom}. First, consider 
a vertex $v\in D_G$. We have 
that $\mathcal{A}_{G,\{v\}}=\{V\setminus\{v\}\}$. 
Now, consider a vertex $v\in S$. 
In this case, $\mathcal{A}_{G,\{v\}}=\{N_G(v), V\setminus N_G[v]\}= 
\{N_G(v), S\setminus N_{\langle S \rangle_G}[v]\}$. In consequence, 
we have that $G$ satisfies $(k,1)$-adjacency anonymity with 

\begin{displaymath} 
\begin{array}{rcl} 
k&=&\displaystyle\min\left\{n-1, \min_{v\in S}\left\{\min\left\{\delta_G(v), 
|S|-\delta_{\langle S \rangle_G}(v)-1\right\}\right\}\right\}\\ 
&=&\displaystyle\min_{v\in S}\left\{\min\left\{\delta_G(v), 
|S|-\delta_{\langle S \rangle_G}(v)-1\right\}\right\}\\
&=&\displaystyle\min\left\{\min_{v\in S}\left\{\delta_G(v), 
|S|-\delta_{\langle S \rangle_G}(v)-1\right\}\right\}\\ 
&=&\displaystyle\min\left\{\min_{v\in S}\left\{\delta_G(v)\right\}, 
\min_{v\in S}\left\{|S|-\delta_{\langle S \rangle_G}(v)-1\right\}\right\}\\ 
&=&\displaystyle\min\left\{\delta(G), |S|-\Delta(\langle S \rangle_G)-1
\right\}
\end{array} 
\end{displaymath} 
\end{proof}

\begin{proposition}\label{prop-k1-anonymity-yesiso-nodom}
Let $G=(V,E)$ be a non-empty graph of order $n$ such that $I_G\neq \emptyset$ 
and let $S=V\setminus I_G$. Then, $G$ satisfies $(k,1)$-adjacency anonymity 
with $$k=\min\left\{\delta(\langle S \rangle_G), n-\Delta(G)-1\right\}.$$
\end{proposition}

\begin{proof}
We follow a reasoning analogous to that of the proofs 
of Propositions~\ref{prop-k1-anonymity-noiso-nodom} 
and~\ref{prop-k1-anonymity-noiso-yesdom}. 
First, consider a vertex $v\in I_G$. We have that 
$\mathcal{A}_{G,\{v\}}=\{V\setminus\{v\}\}$. 
Now, consider a vertex $v\in S$. In this case, 
%$\mathcal{A}_{G,\{v\}}=\{N_G(v)\}=\{N_{\langle S \rangle_G}(v)\}$. 
$\mathcal{A}_{G,\{v\}}=\{N_G(v), V\setminus N_G[v]\}= 
\{N_{\langle S \rangle_G}(v), V\setminus N_G[v]\}$. 
In consequence, we have that $G$ satisfies $(k,1)$-adjacency anonymity with 

\begin{displaymath} 
\begin{array}{rcl} 
k&=&\displaystyle\min\left\{n-1, \min_{v\in S}\left\{ 
\min\left\{\delta_{\langle S \rangle_G}(v), 
n-\delta_G(v)-1\right\}\right\}\right\}\\ 
&=&\displaystyle\min_{v\in S}\left\{\min\left\{\delta_{\langle S \rangle_G}(v), 
n-\delta_G(v)-1\right\}\right\}\\
&=&\displaystyle\min\left\{\min_{v\in S}\left\{\delta_{\langle S \rangle_G}(v), 
n-\delta_G(v)-1\right\}\right\}\\ 
&=&\displaystyle\min\left\{\min_{v\in S}\left\{
\delta_{\langle S \rangle_G}(v)\right\}, 
\min_{v\in S}\left\{n-\delta_G(v)-1\right\}\right\}\\ 
&=&\displaystyle\min\left\{\delta(\langle S \rangle_G), n-\Delta(G)-1\right\}
\end{array} 
\end{displaymath} 
\end{proof}

According to Propositions~\ref{prop-k1-anonymity-noiso-nodom}, 
\ref{prop-k1-anonymity-noiso-yesdom} 
and~\ref{prop-k1-anonymity-yesiso-nodom}, 
in order to enforce $(k,1)$-adjacency anonymity on a graph $G$, it is necessary 
to transform it into a graph $G'$ such that the induced subgraph 
$\langle V(G')\setminus (D_{G'}\cup I_{G'})\rangle_{G'}$ has minimum 
degree greater than or equal to $k$ and maximum degree smaller than 
or equal to $n-k-1$, where $n$ is the order of $G'$. Likewise, in order to 
obtain a $(k,1)$-adjacency anonymous transformation $(G=(V,E),G'=(V',E'))$, 
it is necessary to guarantee that every $v\in V$ such that $1\le\delta_G(v)<k$ 
or $|V|-k-1<\delta_G(v)\le |V|-2$ satisfies $v\notin V'$, or $v\in D_{G'}$,  
or $v\in I_{G'}$, or $k\le \delta_{G'}(v)\le |V'|-k-1$. Based on these facts, 
we propose an algorithm that, given a $(k_0,1)$-adjacency anonymous graph 
$G=(V,E)$ and an integer $k$ such that 
$k_0 < k \le\left\lfloor\frac{|V|-1}{2}\right\rfloor$,  
efficiently obtains a graph $G'=(V,E')$ such that the pair $(G,G')$ 
is a $(k,1)$-adjacency anonymous transformation. The method works by 
performing a series of edge additions and removals upon $G$, as outlined 
in Algorithm~\ref{alg-k-1-anonymous-transf}. 

\begin{algorithm}
\caption{Given a $(k_0,1)$-adjacency anonymous graph $G=(V,E)$ and an integer 
$k\in\left[k_0+1,\left\lfloor\frac{|V|-1}{2}\right\rfloor\right]$, 
obtain a graph $G'=(V,E')$ such that $(G,G')$ 
is a $(k,1)$-adjacency anonymous transformation.} 
\label{alg-k-1-anonymous-transf} 
\begin{algorithmic}[1] 
\STATE $E'\gets E$ 
\STATE\label{st-update-L-1} $L\gets \{v\in V:\;1\le\delta_G(v)<k\}$
\STATE\label{st-update-H-1} $H\gets \{v\in V:\;|V|-k-1<\delta_G(v)\le|V|-2\}$ 
\WHILE {$L\neq\emptyset$}\label{start-L}
\STATE\label{st-search-1} $X\gets \{x\in L:\;\exists y\in L\setminus\{x\} 
\st (x,y)\notin E'\}$
\IF {$X\neq\emptyset$}
\STATE\label{st-search-2} $u\gets {\arg\max}_{x\in X}\{\delta_{G'}(x)\}$
\STATE\label{st-search-3} $Y\gets \{y\in L\setminus\{u\}:\;(u,y)\notin E'\}$
\STATE\label{st-search-4} $v\gets {\arg\max}_{y\in Y}\{\delta_{G'}(y)\}$
\STATE\label{st-last-step-opt} $E'\gets E'\cup\{(u,v)\}$
\ELSE 
\STATE\label{st-search-5} $u\gets {\arg\max}_{x\in L}\{\delta_{G'}(x)\}$
\STATE\label{st-search-6} $Y\gets \{y\in V\setminus L:\;(u,y)\notin E'\}$
\STATE\label{st-search-7} $v\gets {\arg\min}_{y\in Y}\{\delta_{G'}(y)\}$ 
\STATE $E'\gets E'\cup\{(u,v)\}$
\ENDIF
\STATE\label{st-update-L-2} $L\gets \{v\in L:\;1\le\delta_{G'}(v)<k\}$ 
\ENDWHILE \label{end-L} 
\STATE\label{st-update-H-2} $H\gets \{v\in H:\;|V|-k-1<\delta_{G'}(v)\le|V|-2\}$ 
\WHILE {$H\neq\emptyset$}\label{start-H}
\STATE\label{st-search-9} $X\gets \{x\in H:\;\exists y\in H\setminus\{x\} 
\st (x,y)\in E'\}$
\IF {$X\neq\emptyset$}
\STATE\label{st-search-10} $u\gets {\arg\min}_{x\in X}\{\delta_{G'}(x)\}$
\STATE\label{st-search-11} $Y\gets \{y\in H\setminus\{u\}:\;(u,y)\in E'\}$
\STATE\label{st-search-12} $v\gets {\arg\min}_{y\in Y}\{\delta_{G'}(y)\}$
\STATE\label{st-last-step-opt-H} $E'\gets E'\setminus\{(u,v)\}$
\ELSE 
\STATE\label{st-search-13} $u\gets {\arg\min}_{x\in H}\{\delta_{G'}(x)\}$
\STATE\label{st-search-14} $Y\gets \{y\in V\setminus(H\cup L):
\;(u,y)\in E'\}$
\STATE\label{st-search-15} $v\gets {\arg\max}_{y\in Y}\{\delta_{G'}(y)\}$ 
\STATE $E'\gets E'\setminus\{(u,v)\}$
\ENDIF
\STATE\label{st-update-H-3} $H\gets \{v\in H:\;|V|-k-1<\delta_{G'}(v)\le|V|-2\}$  
\ENDWHILE \label{end-H} 
\RETURN $G'$
\end{algorithmic} 
\end{algorithm} 

In Algorithm~\ref{alg-k-1-anonymous-transf}, the sets $L$ and $H$ contain 
the vertices whose degrees are, respectively, smaller and greater 
than required for the privacy requirement to be satisfied (without 
being isolated nor dominant vertices). 
Steps~\ref{st-update-L-1}, \ref{st-search-1}, \ref{st-search-2}, 
\ref{st-search-3}, \ref{st-search-4}, \ref{st-search-5}, \ref{st-search-6}, 
\ref{st-search-7}, \ref{st-update-L-2}, \ref{st-update-H-2}, 
\ref{st-search-9}, \ref{st-search-10}, \ref{st-search-11}, 
\ref{st-search-12}, \ref{st-search-13}, \ref{st-search-14}, \ref{st-search-15} 
and~\ref{st-update-H-3} can be efficiently performed by maintaining 
the elements of $V$ sorted 
by their degree and updating the ordering when necessary. The rationale 
behind the loop in steps~\ref{start-L} to~\ref{end-L} is to first add 
as many edges as possible between pairs of vertices from $L$, 
since every addition of this type increases the degree of two of such vertices. 
When such additions are no longer possible, then 
we add edges linking a vertex $u\in L$ and a vertex $v\notin L$ 
whose degree is as small as possible. The latter condition makes that the 
degree of vertices from $H$ is only increased if there is no vertex 
in $V\setminus H$ to which $u$ can be linked. An analogous idea is applied in 
the loop in steps~\ref{start-H} to~\ref{end-H} to first remove edges joining 
pairs of vertices from $H$, then edges joining a vertex from $H$ to other 
vertex (with the particularity that step~\ref{st-search-14} takes care 
of not making the degree of a vertex from $L$ decrease again), and so on. 
It is worth noting that in real-life social graphs, which are characterised 
by very low densities, and for practical values of $k$, 
steps~\ref{start-H} to~\ref{end-H} are very unlikely to be 
executed. 

Considering the number of modifications performed by the algorithm, the 
best scenario is when all edge additions are done according to 
steps~\ref{st-search-1} to~\ref{st-last-step-opt}, and all edge removals 
are done according to steps~\ref{st-search-9} to~\ref{st-last-step-opt-H}, 
as shown in the following results. 

\begin{theorem}\label{th-bounds-added}
Let $G=(V,E)$ be a $(k_0,1)$-adjacency anonymous social graph 
and let $k\in\left[k_0+1,\left\lfloor\frac{|V|-1}{2}\right\rfloor\right]$. 
The number $t$ of edges added by steps~\ref{start-L} to~\ref{end-L} 
of Algorithm~\ref{alg-k-1-anonymous-transf} satisfies 
\begin{equation}\label{eq-bounds-added}
\left\lceil\frac{\displaystyle\sum_{u\in V,\ 1\le \delta_G(u)<k}k-\delta_G(u)}
{2}\right\rceil \le t \le \sum_{u\in V,\ 1\le \delta_G(u)<k}k-\delta_G(u) 
\end{equation}
\end{theorem}

\begin{proof}
Let $((u_1,v_1),(u_2,v_2),\ldots,(u_t,v_t))$, 
with $(u_i,v_i)\in (V\times V)\setminus E$ for $i\in\{1,\ldots,t\}$, be 
the  sequence of edges added to $G$ by steps~\ref{start-L} to~\ref{end-L} 
of Algorithm~\ref{alg-k-1-anonymous-transf}.  
Let $E_0=E$ and $E_i=E_{i-1}\cup\{(u_i,v_i)\}$, for $i\in\{1,\ldots,t\}$.  
Moreover, for every $i\in\{0,\ldots,t\}$, let $G_i=(V,E_i)$ and 
$L_i=\{v\in L:\;1\le\delta_{G_i}(v)<k\}$. 

After adding the edge $(u_i,v_i)$, we have that 
$\delta_{G_i}(u_i)=\delta_{G_{i-1}}(u_i)+1$ and 
$\delta_{G_i}(v_i)=\delta_{G_{i-1}}(v_i)+1$, whereas 
$\delta_{G_i}(x)=\delta_{G_{i-1}}(x)$ for every 
$x\in V-\{u_i,v_i\}$. 

We define the function 
\begin{displaymath}
missing(G_i)=\sum_{x\in L_i}(k-\delta_{G_i}(x))
\end{displaymath}
which specifies by how much the sum of the degrees of vertices from $L$ 
needs to be increased for $(G,G_i)$ to be a $(k,1)$-adjacency anonymous 
transformation. 
Note that, by the definition of $t$, we have that $missing(G_t)=0$. Moreover, 
$missing(G_0)=\displaystyle\sum_{u\in V,\ 1\le \delta_G(u)<k}k-\delta_G(u)$. 
After adding the edge $(u_i,v_i)$, the following situations are possible:

\begin{enumerate}
\item $u_i,v_i\in L_{i-1}$. In this case, since two vertices from $L_{i-1}$ 
have their degree increased by $1$, we have that  
$missing(G_i)=missing(G_{i-1})-2$. 
\item $u_i\in L_{i-1}$ and $v_i\notin L_{i-1}$, or 
\emph{vice versa}. Here, $missing(G_i)=missing(G_{i-1})-1$. 
\end{enumerate} 

With the previous definitions in mind, we will address the proof 
of the left-hand inequality in Equation~\ref{eq-bounds-added}. 
To that end, we will assume, for the purpose of contradiction, that 
$$t<\left\lceil\frac{\sum_{u\in V,\ 1\le \delta_G(u)<k}k-\delta_G(u)}{2} 
\right\rceil=\left\lceil\frac{missing(G_0)}{2}\right\rceil.$$ 

If $missing(G_0)$ is even, we have that $t<\frac{missing(G_0)}{2}$. 
Given that, in the best case scenario, situation 1 above occurs 
at every iteration of the algorithm, we have  

\begin{displaymath}
\begin{array}{rcl}
missing(G_t)&\geq&missing(G_0)-2t\\
&>&missing(G_0)-2\cdot\frac{missing(G_0)}{2}\\
&=&0
\end{array}
\end{displaymath}
which is a contradiction. 

In a similar manner, if 
$missing(G_0)$ is odd, we have that $t<\frac{missing(G_0)+1}{2}$. Here, 
in the best case scenario, situation 1 above occurs in every iteration, 
except one, so 

\begin{displaymath}
\begin{array}{rcl}
missing(G_t)&\geq&missing(G_0)-2(t-1)-1\\
&>&missing(G_0)-2\left(\frac{missing(G_0)+1}{2}-1\right)-1\\
&=&0
\end{array}
\end{displaymath}
which is also a contradiction. Thus, we can conclude that 
$$t\ge\left\lceil\frac{\sum_{u\in V,\ 1\le \delta_G(u)<k}k-\delta_G(u)}{2} 
\right\rceil.$$ 

The right-hand inequality in Equation~\ref{eq-bounds-added} is trivial, 
given that at least one vertex has its degree increased by $1$ at every 
iteration. The proof is thus complete. 
\end{proof}

The lower and upper bounds provided in Theorem~\ref{th-bounds-added} are 
tight, as exemplified in Figures~\ref{fig-bounds-added}.a) 
and~\ref{fig-bounds-added}.b), respectively. 

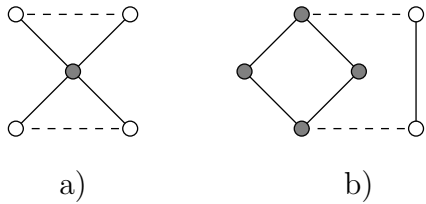
\begin{figure}[!ht]
\begin{center}

\begin{tikzpicture}[inner sep=0.7mm, place/.style={circle,draw=black,
fill=white},xx/.style={circle,draw=black!99, 
fill=black!99},gray1/.style={circle,draw=black!99, 
fill=black!25},gray2/.style={circle,draw=black!99, 
fill=black!50},gray3/.style={circle,draw=black!99,fill=black!75},
transition/.style={rectangle,draw=black!50,fill=black!20,thick}, 
line width=.5pt,scale=0.38]

\coordinate (cv) at (0,0);

\coordinate (v1) at (2,2);
\coordinate (v2) at (-2,2);
\coordinate (v3) at (-2,-2);
\coordinate (v4) at (2,-2);

\draw[black] (v1) -- (v3);
\draw[black] (v2) -- (v4);

\draw[black,dashed] (v1) -- (v2);
\draw[black,dashed] (v3) -- (v4);

\node [gray2] at (cv) {};
\node [place] at (v1) {};
\node [place] at (v2) {};
\node [place] at (v3) {};
\node [place] at (v4) {};

\coordinate [label=center:{a)}] (G1) at (0,-4);

\coordinate (cv2) at (10,0);

\coordinate (v12) at (12,2);
\coordinate (v22) at (8,2);
\coordinate (v32) at (8,-2);
\coordinate (v42) at (12,-2);
\coordinate (v52) at (6,0);

\draw[black] (v22) -- (cv2) -- (v32) -- (v52) -- cycle;
\draw[black] (v12) -- (v42);

\draw[black,dashed] (v12) -- (v22);
\draw[black,dashed] (v32) -- (v42);

\node [gray2] at (cv2) {};
\node [place] at (v12) {};
\node [gray2] at (v22) {};
\node [gray2] at (v32) {};
\node [place] at (v42) {};
\node [gray2] at (v52) {};

\coordinate [label=center:{b)}] (G12) at (10,-4);

\end{tikzpicture}

\end{center}
\caption{Two examples where the number of edges added by steps~\ref{start-L} 
to~\ref{end-L} of Algorithm~\ref{alg-k-1-anonymous-transf} (for $k=2$) 
reaches the (a) lower and (b) upper bounds of the inequalities 
in Equation~\ref{eq-bounds-added}. 
In both cases, dashed lines indicate the edges added by the algorithm.} 
\label{fig-bounds-added}
\end{figure}

%Following a reasoning analogous to the one applied 
%in Theorem~\ref{th-bounds-added}, 
The next result describes the number 
$t'$ of edges removed by steps~\ref{start-H} to~\ref{end-H} 
of Algorithm~\ref{alg-k-1-anonymous-transf}. 

\begin{theorem}\label{th-bounds-removed}  
Let $G=(V,E)$ be a $(k_0,1)$-adjacency anonymous social graph 
and let $k\in\left[k_0+1,\left\lfloor\frac{|V|-1}{2}\right\rfloor\right]$. 
Let $G_t$ be the graph obtained from $G$ after executing steps~\ref{start-L} 
to~\ref{end-L} of Algorithm~\ref{alg-k-1-anonymous-transf}. The number 
$t'$ of edges removed by steps~\ref{start-H} to~\ref{end-H} 
of Algorithm~\ref{alg-k-1-anonymous-transf} satisfies 
\begin{equation}\label{eq-lower-bound-removed} 
t'\ge\left\lceil\frac{\displaystyle\sum_{u\in V,\ |V|-k-1<\delta_G(u)\le|V|-2} 
\left[k-\left(|V|-\delta_{G_t}(u)-1\right)\right]}{2}\right\rceil 
\end{equation} 
and 
\begin{equation}\label{eq-upper-bound-removed} 
t' \le \sum_{u\in V,\ 
|V|-k-1<\delta_G(u)\le|V|-2} \left[k-\left(|V|-\delta_{G_t}(u)-1\right)\right] 
\end{equation} 
\end{theorem} 

\begin{proof}
We will follow a reasoning analogous to the one applied in the proof of 
Theorem~\ref{th-bounds-added}. 
Let $((u_1,v_1),(u_2,v_2),\ldots,(u_{t'},v_{t'}))$, 
with $(u_i,v_i)\in E_t=E(G_t)$ for $i\in\{1,\ldots,t'\}$, be 
the sequence of edges removed from $G_t$ by steps~\ref{start-H} to~\ref{end-H} 
of Algorithm~\ref{alg-k-1-anonymous-transf}.  
Let $E_{t+i}=E_{t+i-1}\setminus\{(u_i,v_i)\}$, for $i\in\{1,\ldots,t'\}$.  
Moreover, for every $i\in\{1,\ldots,t'\}$, let $G_{t+i}=(V,E_{t+i})$ and 
$H_i=\{v\in H:\;|V|-k-1<\delta_{G_{t+i}}(v)\le|V|-2\}$. 

After removing the edge $(u_i,v_i)$, we have that 
$\delta_{G_{t+i}}(u_i)=\delta_{G_{t+i-1}}(u_i)-1$ and 
$\delta_{G_{t+i}}(v_i)=\delta_{G_{t+i-1}}(v_i)-1$, whereas 
$\delta_{G_{t+i}}(x)=\delta_{G_{t-i-1}}(x)$ for every 
$x\in V-\{u_i,v_i\}$. 

Now we introduce the function 
\begin{displaymath}
excess(G_{t+i})=\sum_{x\in H_i}\left[k-(|V|-\delta_{G_{t+i}}(x)-1)\right]. 
\end{displaymath}

In a manner analogous to the proof of Theorem~\ref{th-bounds-added}, 
we have that by definition $excess(G_{t+t'})=0$ and 
$excess(G_t)=\displaystyle\sum_{u\in V,\ |V|-k-1<\delta_G(u)<|V|-2} 
\left[k-\left(|V|-\delta_G(u)-1\right)\right]$. Additionally, after removing 
the edge $(u_i,v_i)$, the following situations are possible: 

\begin{enumerate}
\item $u_i,v_i\in H_{i-1}$. In this case, since two vertices from $H_{i-1}$ 
have their degree decreased by $1$, we have that  
$excess(G_{t+i})=excess(G_{t+i-1})-2$. 
\item $u_i\in H_{i-1}$ and $v_i\notin H_{i-1}$, or 
\emph{vice versa}. Here, $excess(G_{t+i})=excess(G_{t+i-1})-1$. 
\end{enumerate} 

Now, to address the proof of the inequality 
in Equation~\ref{eq-lower-bound-removed}, we assume, for the purpose 
of contradiction, that 
$$t'<\left\lceil\frac{\displaystyle\sum_{u\in V,\ |V|-k-1<\delta_G(u)\le|V|-2} 
\left[k-\left(|V|-\delta_{G_t}(u)-1\right)\right]}{2}\right\rceil
=\left\lceil\frac{excess(G_t)}{2}\right\rceil.$$ 

In consequence, if $excess(G_t)$ is even, we have 
\begin{displaymath}
\begin{array}{rcl}
excess(G_{t+t'})&\geq&excess(G_t)-2t'\\
&>&excess(G_t)-2\cdot\frac{excess(G_t)}{2}\\
&=&0
\end{array}
\end{displaymath}
which is a contradiction, whereas in the case that $excess(G_t)$ is odd we have 
\begin{displaymath}
\begin{array}{rcl}
excess(G_{t+t'})&\geq&excess(G_t)-2(t'-1)-1\\
&>&excess(G_t)-2\left(\frac{excess(G_t)+1}{2}-1\right)-1\\
&=&0
\end{array}
\end{displaymath}
which is also a contradiction, so we can conclude 
that Equation~\ref{eq-lower-bound-removed} holds. 
As in Theorem~\ref{th-bounds-added}, the upper bound 
(Equation~\ref{eq-upper-bound-removed}) is trivial. 
\end{proof}

%\YR{Can the best (worst) cases of Theorems~\ref{th-bounds-added} 
%and~\ref{th-bounds-removed} occur at the same time?} 

\section{$(2, \ell)$-adjacency anonymous transformations} 
\label{sect-2-ell-trans} 

In Algorithm~\ref{alg-k-1-anonymous-transf}, the fact that a vertex $v$ 
satisfies $v\in L$ means that the equivalence class composed by the vertices 
having adjacency representation $(1)$ with respect to 
the set $\{v\}$ in the original graph $G$ is not empty and its cardinality is 
smaller than~ $k$. 
Likewise, the fact that $v\in H$ means that the equivalence class composed 
by the vertices having adjacency representation $(2)$ with respect to the set 
$\{v\}$ in $G$ is not empty and its cardinality  
is smaller than $k$. To a limited extent, a strategy similar to the one applied 
in Algorithm~\ref{alg-k-1-anonymous-transf} can be used to obtain 
$(k,\ell)$-adjacency anonymous transformations with $\ell>1$. 
%, as exemplified 
%in Appendix~\ref{app-alg-k-2-anonymous-transf}, where we sketch an algorithm 
%following the philosophy of Algorithm~\ref{alg-k-1-anonymous-transf} 
%for obtaining $(k,2)$-adjacency anonymous transformations. In this case, 
For example, for $\ell=2$, 
in addition to the sets $L,H\subset V(G)$, we would consider the sets 
$P_{11}, P_{12}, P_{21}, P_{22}\subseteq V(G)\times V(G)$, where 
$(u,v)\in P_{ij}$, $i,j\in\{1,2\}$, means that the the equivalence class 
composed by the vertices having adjacency representation $(i,j)$ 
with respect to $(u,v)$ in $G$ is not empty and its cardinality 
is smaller than $k$. 
Thus, the algorithm would work by executing the necessary edge set editions 
to increase the cardinalities of these equivalence classes or, alternatively, 
to empty them. However, it is impractical 
to use this philosophy in the general case, as it entails designing 
a different, highly casuistic algorithm for every different value of $\ell$. 

For the general case, we have devised a greedy edge-addition-based 
method that, for small values of $\ell$, allows to obtain 
$(2,\ell)$-adjacency anonymous transformations. Given a graph $G=(V,E)$, 
the method starts by computing all $1$-adjacency antiresolving sets of $G$ 
of cardinality at most $\ell$. Then, edges are iteratively 
added until obtaining a graph $G'=(V,E\cup E')$ such that $(G,G')$ 
is a $(2,\ell)$-adjacency anonymous transformation. The critical aspect 
of the proposed framework is how to determine an appropriate order for 
adding edges. 
%in order to speed-up convergence and minimise utility loss. 

In order to describe the proposed method, we will first introduce the following 
results, which characterise the sets of edges whose addition 
to a graph $G$ may modify the set of $1$-adjacency antiresolving sets. 
In what follows, we will use the notation $\mathcal{S}_{G,\ell}$ for 
the set of $1$-adjacency antiresolving sets of a graph $G$ having cardinality 
smaller than or equal to $\ell$. 

\begin{remark}\label{rm-useless-non-edges}
Let $G=(V,E)$ be a social graph, $u,v\in V$ a pair of vertices 
of $G$ such that $(u,v)\notin E$, and $G'=(V,E\cup\{(u,v)\})$. 
If, for every $S\in\mathcal{S}_{G,\ell}$, 
either $S\cap\{u,v\}=\{u,v\}$ or $S\cap\{u,v\}=\emptyset$, then 
$\mathcal{S}_{G,\ell}\setminus\mathcal{S}_{G',\ell}=\mathcal{S}_{G,\ell}$. 
\end{remark}

\begin{proof}
The result follows directly from the fact that, for every 
$S\in\mathcal{S}_{G,\ell}$ and every $x\in V\setminus S$, 
we have that $a_G(x\ | \ S)=a_{G'}(x\ | \ S)$, so $S\in\mathcal{S}_{G',\ell}$.
\end{proof} 

Algorithm~\ref{alg-2-ell-anonymous-transf} describes the edge-addition method. 
First, we use Remark~\ref{rm-useless-non-edges} to discard candidate vertex 
pairs $(u,v)$ that are known not to cause any $1$-adjacency antiresolving set 
of the current graph $G=(V,E)$ to become a $k$-adjacency antiresolving set 
of $G'=(V,E\cup\{(u,v)\})$ with $k>1$. Then, every remaining candidate pair 
$(u,v)$ is scored as follows: 
$$score((u,v))=|\{S:\;S\in\mathcal{S}_{G,\ell},u\in S,[v]_{S}^{G}=\{v\}\}|.$$  
where $[v]_{S}^{G}$ represents the equivalence class of $v$ 
in $\mathcal{A}_{G,S}$. 
In other words, we consider the number of times the candidate pair  
would modify the fingerprint of a uniquely identifiable vertex with respect 
to a $1$-adjacency antiresolving set of $G$. 
The intuition behind this heuristics is that the larger the number of times 
that the pair $(u,v)$ is found in this situation, the larger the likelihood 
that adding the edge $(u,v)$ will result in making some vertex set stop being 
$1$-adjacency antiresolving. At every iteration, the current perturbed graph 
$G'=(V,E')$ is transformed into the graph $G''=(V,E'\cup\{(u,v)\})$, 
where $(u,v)$ is the best-scored candidate addition satisfying 
$(\mathcal{S}_{G'',\ell}\setminus\mathcal{S}_{G',\ell})\cap\mathcal{S}_{G,\ell} 
=\emptyset$. 

\begin{algorithm}
\caption{Given a graph $G=(V,E)$ and a positive integer $\ell\ge 2$, 
obtain a graph $G'=(V,E')$ such that $(G,G')$ is a $(2,\ell)$-anonymous 
transformation.} 
\label{alg-2-ell-anonymous-transf} 
\begin{algorithmic}[1] 
%\STATE $C\gets (V\times V)\setminus E$ 
\STATE $C\gets \emptyset$ 
\STATE Compute $\mathcal{S}_{G,\ell}$
\STATE $C\gets \emptyset$ 
\FOR {$(u,v)\in (V\times V)\setminus E$} 
\FOR {$S \in \mathcal{S}_{G,\ell}$} 
\IF {$|S\cap(u,v)|=1$} 
\STATE $C\gets C\cup\{(u,v)\}$ 
\BREAK 
\ENDIF 
\ENDFOR 
\ENDFOR 
%\STATE Remove from $C$ useless candidates as characterised in 
%Remark~\ref{rm-useless-non-edges} 
\STATE Decrementally sort $C$ by 
$score((u,v))=|\{S:\;S\in\mathcal{S}_{G,\ell},u\in S,[v]_{S}^{G}=\{v\}\}|$
\STATE $t\gets 0$
\STATE $E_0\gets E$ 
\WHILE {$\mathcal{S}_{G,\ell}\setminus\mathcal{S}_{G_t,\ell}\neq\emptyset$} 
\STATE\label{st-check-back} $(u,v)={\arg\max}_{score(u,v)}\{(u,v)\in C:\; 
(\mathcal{S}_{(V,E_t\cup\{(u,v)\}),\ell}\setminus\mathcal{S}_{G_t,\ell}) 
\cap\mathcal{S}_{G,\ell}=\emptyset\}$ 
\STATE $E_{t+1}\gets E_t\cup\{(u,v)\}$ 
\STATE $C\gets C\setminus\{(u,v)\}$ 
\STATE $t\gets t+1$ 
\ENDWHILE
\RETURN $G_t$
\end{algorithmic} 
\end{algorithm} 

The asymptotic time complexity of Algorithm~\ref{alg-2-ell-anonymous-transf} 
is dominated by the computation and traversals of $\mathcal{S}_{G,\ell}$, 
which is $O(2^n)$ in the general case. However, for small values of $\ell$, 
these computations can be done in $O(n^{\ell})$ time. 
As we discussed before, active adversaries can only insert a limited amount 
of sybil nodes in the network without being detected, so the capacity 
of protecting the graph against such adversaries results in an important 
privacy increase. As an aid to speed-up the algorithm, the following result 
shows how the number of verifications to perform in evaluating the condition 
at step~\ref{st-check-back} can be largely reduced. 

\begin{theorem}\label{theo-limit-cross-check} 
Let $G=(V,E)$ be a social graph, $u,v\in V$ a pair of vertices 
of $G$ such that $(u,v)\notin E$, and $G'=(V,E\cup\{(u,v)\})$. 
Let $S\subseteq V$ such that it is a $k$-adjacency antiresolving set 
of $G$, with $k\ge 2$, and a $1$-adjacency antiresolving set 
of $G'$. If there exists $w\in S$ such that $d_G(u,w)>2$ and $d_G(v,w)>2$, 
then for every $x\in V\setminus S$ such that $[x]_S^{G'}=\{x\}$ it holds 
that $[x]_{S\setminus\{w\}}^{G'}=\{x\}$ and $|[x]_{S\setminus\{w\}}^{G}|>1$. 
\end{theorem} 

\begin{proof} 
Consider a graph $G=(V,E)$, a pair of vertices $u,v\in V$ such that 
$(u,v)\notin E$ and a set $S\subseteq V$ satisfying the premises of 
Theorem~\ref{theo-limit-cross-check}. 
Also consider a vertex $w\in S$ such that $d_G(u,w)>2$ and $d_G(v,w)>2$ 
and a vertex $x\in V\setminus S$ such that $[x]_S^{G'}=\{x\}$. By the 
definition of $S$ we have that $|[x]_S^G|>1$, so there exists 
$y\in V\setminus(S\cup\{x\})$ such that $a_G(x|\;S)=a_G(y|\;S)$ 
and $a_{G'}(x|\;S)\neq a_{G'}(y|\;S)$. Suppose, for the purpose 
of contradiction, that $a_{G'}(w,x)\neq a_{G'}(w,y)$. Then, 
since $w\neq u$ and $w\neq v$, we have that $a_{G}(w,x)\neq a_{G}(w,y)$, 
which contradicts the fact that $a_G(x|\;S)=a_G(y|\;S)$. 
Therefore, we have that $a_{G'}(w,x)=a_{G'}(w,y)$, which 
implies that $a_G(x|\;S\setminus\{w\})=a_G(y|\;S\setminus\{w\})$ 
and $a_{G'}(x|\;S\setminus\{w\})\neq a_{G'}(y|\;S\setminus\{w\})$. 
Since the only difference between $G$ and $G'$ is the addition of the 
edge $(u,v)$, we conclude that either $x=u$ and $v\in S$, or 
\emph{vice versa}. For the remainder of this proof, we will assume 
$x=u$ and $v\in S$ without loss of generality. 

We will now proceed by reduction to absurdity. 
To that end, we will assume that 
$|[x]_{S\setminus\{w\}}^{G'}|>1$ or $[x]_{S\setminus\{w\}}^{G}=\{x\}$. 
It is simple to see that $[x]_{S\setminus\{w\}}^{G}=\{x\}$ contradicts 
the fact that $S$ is a $k$-adjacency antiresolving set of $G$ with $k\ge 2$. 
Therefore, in what follows we will focus on the assumption that 
$|[x]_{S\setminus\{w\}}^{G'}|>1$. 
In this case, there exists 
$y\in V\setminus (S\cup\{x\})$ such that $a_{G'}(x|\;S\setminus\{w\})= 
a_{G'}(y|\;S\setminus\{w\})$. Since $[x]_S^{G'}=\{x\}$, we have that 
$a_{G'}(w,x)\neq a_{G'}(w,y)$. Moreover, since $(w,x)\neq(u,v)$ 
and $(w,y)\neq(u,v)$, we have that $a_{G}(w,x)=a_{G'}(w,x)$ 
and $a_{G}(w,y)=a_{G'}(w,y)$. Thus, it holds that 
$a_{G'}(w,x)\neq a_{G'}(w,y)\implies a_G(w,x)\neq a_G(w,y)$, which entails 
$$(w\sim_{_G}x\wedge w\nsim_{_G}y)\vee(w\nsim_{_G}x\wedge w\sim_{_G}y)$$
Since $x=u$, we have that $w\nsim_{_G}x$, because $d_G(u,w)>2$, so  
$w\sim_{_G}y$. Moreover, since $u\sim_{_{G'}}v$ and 
$a_{G'}(v,x)=a_{G'}(v,y)$, then $v\sim_{_{G}}y$, which implies $d_G(w,v)=2$, 
again a contradiction. This concludes the proof. 
\end{proof} 

According to Theorem~\ref{theo-limit-cross-check}, when verifying if 
the addition of an edge $(u,v)$ causes some $k$-adjacency antiresolving 
set (with $k\ge 2$) of $G_t$ to become a $1$-adjacency antiresolving set 
of $G'_t=(V(G_t),E(G_t)\cup\{(u,v)\})$, it suffices 
to analyse those sets $S\in\mathcal{S}_{G_t,\ell}$ such that some 
$w\in S$ satisfies $d_{G_t}(u,w)\le 2$ or $d_{G_t}(v,w)\le 2$.

\section{Concluding remarks} 
\label{sect-conclusions}

In this paper, we have re-visited the notion of $(k,\ell)$-anonymity, which 
quantifies the privacy level of a social graph in the presence of active 
adversaries. Firstly, we have introduced the notion of $(k,\ell)$-anonymous 
transformations, which allow to reduce the amount of perturbation needed 
to protect a social graph from an active attack. Secondly, we have 
critically assessed the assumptions posed by $(k,\ell)$-anonymity on the 
adversary capabilities. Judging that it is unrealistic to assume that 
an adversary will be able to control all distances between a set of sybil 
nodes and every other vertex of the social graph, we introduced a new 
privacy property: $(k,\ell)$-adjacency anonymity, which accounts 
for adversaries who control the connection patterns with the neighbours 
of the sybil nodes. Finally, combining the two previous ideas, we have 
introduced $(k,\ell)$-adjacency anonymous transformations, which are able 
to protect a social graph from active adversaries levaraging up to $\ell$ 
sybil nodes and constructing fingerprints based on the connection patterns 
between victims and sybil nodes. We proposed two algorithms: one for 
obtaining $(k,1)$-adjacency anonymous transformations for arbitrary values 
of $k$, and another for obtaining $(2,\ell)$-adjacency anonymous 
transformations for small values of $\ell$. The first algorithm is efficient 
and the number of changes introduced in the graph is bounded. On the other 
hand, there is still room for improvement in the second method, especially 
concerning the order in which graph perturbations are applied. We are 
currently using a greedy heuristic to guide the edge-addition process. 
We will evaluate the convenience of this heuristic, and explore the use
of meta-heuristics  such as genetic algorithms and ant-colony optimisation. 

\bibliographystyle{elsart-num-sort} 
%\bibliography{all-refs} 

%\appendix
%\section{An algorithm for obtaining $(k,2)$-adjacency anonymous transformations} 
%\label{app-alg-k-2-anonymous-transf} 

\end{document}